\documentclass[10pt, twocolumn,twoside]{IEEEtran}

\usepackage{cite}
\usepackage{amsmath}
\usepackage{amsthm}
\usepackage{amssymb}
\usepackage{algorithmic}
\usepackage{array}
\usepackage[final]{graphicx} 
\usepackage{pgfplots} 
\usepackage{url}
\usepackage{upgreek}
\usepackage{mathrsfs}
\usepackage[caption=false]{subfig}
\DeclareCaptionLabelSeparator{periodspace}{.\quad}
\usepackage[final]{graphicx}
\usepackage{multirow}
\usepackage{rotating}
\usepackage{color}
\usepackage[utf8]{inputenc}

\usepackage{enumerate}

\usepackage{textcomp}
\usepackage{oldgerm}
\usepackage{bm} 

\newtheorem{proposition}{Proposition}
\newtheorem{definition}{Definition}
\newtheorem{theorem}{Theorem}
\newtheorem{corollary}{Corollary}
\newtheorem{lemma}{Lemma}

\def\RR{\mathbb{R}}

\def\ZZ{\mathbb{Z}}

\def\V#1{{\bm #1}}          
\def\M#1{{{\mathbf{#1}}}}   

\def\ud{\mathrm{d}}
\def\ue{\mathrm{e}}

\def\uj{\mathrm{j}}

\def\Re{\mathop{\mathrm{Re}}\nolimits}

\def\sfrac#1#2{\hbox{$\frac{#1}{#2}$}} 
\def\HH{\mathrm{H}} 

\newcommand{\norm}[1]{\left\lVert#1\right\rVert}
\newcommand{\abs}[1]{\left\vert #1\right\vert }

\definecolor{indigo}{rgb}{.4,.07,.48}
\long\def\ED#1{{\color{black}#1}}

\begin{document}
\graphicspath{ {./figures/} }


\title{Angular Accuracy of Steerable Feature Detectors}
\author{Zsuzsanna~P\"{u}sp\"{o}ki,
{Arash Amini}, {Julien Fageot}, {John Paul Ward},
             and~Michael~Unser,~\IEEEmembership{Fellow,~IEEE}
\thanks{The research leading to these results has received funding from the European Research Council under the European Union's Seventh Framework Programme (FP7/2007-2013)
ERC grant agreement N\textdegree $ 267439$ and the Hasler Foundation.}
\thanks{Z. P\"usp\"oki, J. Fageot and M. Unser are with the Biomedical Imaging Group, \'Ecole polytechnique f\'ed\'erale de Lausanne (EPFL), Station 17, CH--1015 Lausanne VD, Switzerland. A. Amini is with the department of Electrical Engineering and Advanced Communications Research Institute (ACRI), Sharif University of Technology, Tehran, Iran. J.P. Ward contributed to this work while affiliated each of the following: Department of Mathematics, University of Central Florida, Orlando, FL 32816, USA and Department of Mathematics, North Carolina A\&T State University, Greensboro, NC 27411, USA.}}

\markboth{Angular Accuracy of Steerable Feature Detectors}
{P\"{u}sp\"{o}ki, Amini, Fageot, Ward, and Unser}

\maketitle

\ifCLASSOPTIONpeerreview
\begin{center}
 \vspace{-0.6cm}
 Authors' contact information:\\
 \vspace{0.2cm}
 \'{E}cole polytechnique f\'{e}d\'{e}rale de Lausanne\\
 Biomedical Imaging Group\\
 CH-1015 Lausanne, Switzerland.\\
 Tel: +41(0)216935136, Fax: +41(0)216933701.\\
 Email: \url{{zsuzsanna.puspoki, michael.unser}@epfl.ch}\\
 Web: \url{http://bigwww.epfl.ch/}\\
 \vspace{0.2cm}
\end{center}
\fi

\begin{abstract}
\boldmath

The detection of landmarks or patterns is of interest for extracting features in biological images. Hence, algorithms for finding these keypoints have been extensively investigated in the literature, and their 
localization and detection properties are well known.  In this paper, we study the complementary topic of local orientation estimation, which has not received similar attention.    
Simply stated, the problem that we address is the following: estimate the angle of rotation of a pattern with steerable filters centered at the same location, where the image is corrupted by colored isotropic Gaussian noise. 
For this problem, we use a statistical framework based on the Cram\'er-Rao lower bound (CRLB) that sets a fundamental limit on the accuracy of the corresponding class of estimators. 
We propose a scheme to measure the performance of estimators based on steerable filters (as a lower bound), while considering the connection to maximum likelihood estimation. Beyond the general results, we analyze the asymptotic behaviour of the lower bound in terms of the order of steerablility and propose an optimal subset of components that minimizes the bound. 
We define a mechanism for selecting optimal subspaces of the span of the detectors. These are characterized by the most relevant angular frequencies.
Finally, we project our template to a basis of  steerable functions and experimentally show that the prediction accuracy achieves the predicted CRLB.  As an extension, we also consider steerable wavelet detectors.

\end{abstract}

\ifCLASSOPTIONjournal
\fi

\IEEEpeerreviewmaketitle

\section{Introduction}\label{sec:Introduction}

Steerable feature detectors \cite{OrientationZs} are popular tools for the quantitative analysis of images where the information that corresponds to the notions of rotation, directionality, and orientation is a key component.
Their application area is continuously growing: 
from nano to macro, 
from biomedical imaging 
\cite{patton_retinal_2006} 
to astronomy 
\cite{schuh_comparative_2014}, 
from material sciences 
\cite{Dan} 
to aerial and satellite imaging \cite{4378557}, 
and so on.
The analysis of local directional patterns also includes the detection of ridges and junctions of any order; applications can be found in  
\cite{OlivoM2010, lung, Xia, Maire, Drawing}. 

Typically, concerning steerable detectors, a two-step algorithm is used for identifying features of interest and determine their orientation. The first step is to find their location in an image. Then, the second step is to steer the detector centred at the given keypoints to identify their exact orientation. In this paper, we focus on the second step only, studying the problem of extracting the local orientation of patterns in images. These local patterns, in particular, include junctions or crossing points with $N$-fold angular symmetry of any order.

In this paper, we examine the behavior of steerable filters in orientation detection, providing a lower bound on the performance of any detector that is in the span of a given family of steerable functions. In particular, we focus on steerable filters that are linear combinations of circular harmonics with given radial profiles. 
We formulate the problem in a statistical framework using the Cram\'er-Rao lower bound (CRLB) that  gives a limit on the error of the estimation. We investigate the connection to maximum likelihood estimation and show that, under some assumptions, our measurement functions correspond to the maximum likelihood estimator. 
We define the best subset of harmonics as the one that minimizes the CRLB for a given class of detectors and we provide a criterion for obtaining it.   
We also analyze the asymptotic behaviour of the lower bound in terms of the number of harmonics. This gives a theoretical limit on the precision of the estimation one can obtain with a given class of estimators.   
We experimentally show that the actual performance of the given steerable detectors follows the predicted theoretical bounds. 

\subsection{State of the Art}

Classical methods to detect orientations are based on gradient information (\textit{e.g.}, Canny edge detector \cite{Canny1986}), on directional derivatives \cite{Xia}, and on the structure tensor \cite{Jahne}. Variations of the latter method can be found in \cite{Koethe2003} and \cite{Bigun.etal2004}.
To capture higher-order directional structures, the Hessian and higher-order derivatives can be used \cite{Foerstner86, Bigun87}.
While simple and computationally efficient, these methods have drawbacks: they only take into account one specific scale and the estimation of the orientation can be overly sensitive to noise.

Alternatively, directional pattern matching is also commonly used. 
It often relies on the discretization of the orientation and thus demands a tradeoff between accurate results and computational cost. 
An important exception to this are the steerable filters, where one may perform arbitrary (continuous) rotations and optimizations with a substantially reduced computational overhead. The basics of steerability were formulated by Freeman and Adelson in the early nineties \cite{freeman1990steerable} and developed further by Perona \cite{perona1992steerable}, Portilla and Simoncelli \cite{PortillaParametric}, and Ward and Unser \cite{ward2013harmonic}. Applications of steerable filters were presented in \cite{JacobSteer, karssemeijer, Schmitter}. In \cite{MUE12a}, multisteerable filters were designed and adjusted to the patterns of interest to determine the precise angular distribution of coinciding branches. 
There, the detection and classification of polar-separable patterns (including junctions) rely on a classical structure-tensor scheme, complemented by the  multisteerable filters.
\ED{The authors are not aware of any previous studies that thoroughly address the angular accuracy of steerable feature detectors.}

Over the past decades, multiresolution methods have attracted a lot of interest. In particular, the design of directional wavelets has become a popular way to construct robust filters. This approach has made it possible to process oriented features independently at different scales.
In \cite{NicolasSIAM}, Unser and Chenouard proposed a unifying parametric framework for 2D steerable wavelet transforms. The goal of \cite{NicolasSIAM} was to combine steerability with tight wavelet frames and propose a general scheme to design such wavelets.
In \cite{ZsSteerable} and \cite{PuspokiISBI}, the authors designed wavelets that can serve as a basis to detect features such as junctions and local symmetry points.
An application of steerable wavelets for texture learning was presented in \cite{Depeursinge}.
Here, steerable templates are represented in the Fourier domain using circular harmonics. The order of the detector is given by the number of harmonics, which  also affects the performance of the detection. 

\subsection{Roadmap}

The paper is organized as follows: 
In Section \ref{EP}, we provide a description of our directional pattern model and formulate the estimation problem. 
In Section \ref{measframework}, we discuss our estimation strategy and propose a reference class of estimators. 
In Section \ref{sec:CRLB}, we identify the probability law of the measurement vector and introduce the Cram\'er-Rao lower bound (CRLB).
In Section \ref{sec:MLE}, we consider the connection to maximum likelihood estimation. 
In Section \ref{sec:crlb1}, we provide the CRLB for orientation estimation based on steerable filters. 
We consider the case when the number of angular frequencies is finite (Section \ref{finite}) and also the asymptotic case for a fixed radial profile (Section \ref{asymptotic}). 
In Section \ref{bestN}, we examine the question of how to choose the best subset of harmonics to achieve the lowest CRLB. 
Section \ref{experiments} contains experiments related to the CRLB.
In Section \ref{sec:crlb2}, we extend our results to wavelets. 


\section{The Estimation Problem at a Glance}
\label{epg}

\subsection{Notations}

We use $f(\V x)$ with $\V x \in \RR^2$ and $f(r,\theta)$ with $r \in \RR^+$, $\theta \in [0,2\pi) $ to denote the Cartesian and polar representations of the same 2D function $f$, respectively. Corresponding notations in the Fourier domain are $\hat{f}(\V \omega)$ and  $\hat{f}(\omega,\varphi)$ with $\V \omega \in \RR^2$ and $\omega \in \RR^+$, $\varphi \in [0,2\pi) $.
The Fourier transform of a Lebesgue-integrable, finite-energy function $f \in L_1 \left(\mathbb{R}^2 \right) \cap L_2 \left(\mathbb{R}^2 \right) $ is denoted by $\mathcal{F }\{ f \} = \hat{f}$ and computed according to
\begin{equation}
\hat{f}(\bm{\omega})
=
\int_{\mathbb{R}^2} f(\bm{x})
\ue^{ -\uj\left<\bm{x},\bm{\omega}\right>} {\rm d}\bm{x},
\end{equation} 
where $\langle \bm{x}, \bm{\omega} \rangle = x_1\omega_1 + x_2\omega_2$ is the usual scalar product on $\RR^2$. 

The 2D matrix of rotation by $\theta_0$ is denoted by 
\begin{equation}
\M R_{\theta_0} = \begin{pmatrix}
\cos(\theta_0) &  -\sin(\theta_0)\\
\sin(\theta_0) &  \phantom{+}\cos(\theta_0)
\end{pmatrix}.
\end{equation}

\subsection{The Estimation Problem}
\label{EP}

Our interest in this paper lies in detecting a rotated pattern $J(\V x)$ from a noisy signal of the form 
\begin{align} \label{eq:imagemodel}
I(\V x ) = J(\M R_{- \theta^*} (\V x- \V x_0)) + S(\V x)
\end{align}
using a steerable filter bank. 
Here, $J$ denotes the general shape of the pattern of interest and $J(\M R_{- \theta^*} (\V x - \V x_0))$ its rotated version around location $\bm{x}_0$ with an unknown angle $\theta^*$. 
$S$ is the background signal, modeled as the realization of an isotropic Gaussian self-similar random field. The motivation behind this choice is that the power spectrum of many natural images is isotropic with an $1/ \lVert \V\omega \rVert^\gamma$ type of decay, which is consistent with long-range dependencies \cite{Vehel, Pentland}. Also, this model of background signal fits fluorescence microscopy images well \cite{SageTracking}, which is relevant to many practical applications of orientation-estimation methods. Further details on our noise model are given in Appendix \ref{sec:noise}.

We shall perform the correlations between measurement filters and the image $I$ in a sliding fashion, and we are interested in determining the performance of this detection method in terms of angular accuracy. To that end, we shall only examine the position where the detector hits the target. Thus, without loss of generality, we set $\V{x}_0 = \V{0}$ in \eqref{eq:imagemodel}.

We analyze the image $I$ through a family of measurement filters $\xi_{\alpha}$ parameterized by a (multi)-index $\alpha$. 
The measurements at location $\V{x}_0 = \V{0}$ are then given by
\begin{align}  \label{eq:qalpha}
q_{\alpha}   = \langle I  , \xi_{\alpha}  \rangle  
   = \langle J(\M R_{- \theta^*} \cdot ) , \xi_{\alpha} \rangle + \langle S, \xi_{\alpha} \rangle.
\end{align}
We also set $u_\alpha = \langle J , \xi_\alpha \rangle$ and $s_\alpha = \langle S , \xi_\alpha \rangle$. 

We distinguish three different cases related to the construction of $q_{\alpha}$ with $\hat{\xi}_{\alpha}(\omega)= \hat{h}_{\alpha}(\omega) \ue^{\uj n_{\alpha}\varphi}$ (cf. Table \ref{tab:q_alpha}). In this paper, we address the question of angular accuracy of measurement filters constructed as conventional as well as wavelet detectors. 
\begin{table}[!t]
\centering
\caption{Construction of measurement functions.}
\vspace{0.05 in}
\begin{tabular}{c l c  l c}
\hline
\hline
$q_{\alpha}$ & radial profile & harmonics\\
\hline
Conventional detector \\$\alpha =n \in\ZZ $ & $\hat{\xi}_n(\omega) = \hat{h}(\omega)$ fixed & $\ue^{\uj n\theta}$\\
\hline
Wavelet detector \\$\alpha =(n,i)\in\ZZ^2$ & $\hat{\xi}_{(n,i)}(\omega) = 2^i\hat{h}(2^i\omega)$ fixed & $\ue^{\uj n\theta}$\\
\hline
General detector\\ $\alpha \in\ZZ^d $ & $\hat{\xi}_\alpha(\omega) = \hat{h}(\omega)$ adaptive & $\ue^{\uj n_\alpha \theta}$\\
\hline
\hline
\end{tabular}
\label{tab:q_alpha}
\end{table}

The local orientation angle $\theta^*$ is estimated from the vector of measurements $\V q = (q_\alpha)$. Any estimator based on this framework is then a mapping $E$ that takes the measurements  $\V q$ and returns an estimate $\tilde{\theta}$ of $\theta^*$. 

\subsection{Steerable Filterbanks and Estimation Strategy}
\label{measframework}

We perform the estimation of the unknown angle $\theta^*$ by selecting a suitable filter $\xi$ that is a linear combination of the measurement filters $\xi_\alpha$, and by selecting $\tilde{\theta}$ as the solution of
\begin{equation}\label{eq:tildetheta} 
	\tilde{\theta}= {\text{argmax}}_{\theta_0 \in [0,2\pi)} 
\left\langle I, \xi (   \M R_{\theta_0} \cdot )  \right\rangle.
\end{equation}
In order to define an estimator that meets the requirements of Section \ref{EP},   $\xi$ needs to satisfy three properties.
\begin{itemize}
	\item The filter $\xi$ has to be a good approximation of the pattern of interest $J$, with the consequence that \eqref{eq:tildetheta} corresponds to the detection of this pattern at the correct orientation. \ED{We note that angular (quasi) symmetries in $J$ might lead to mis-detection of the angle. }
	\item The filter $\xi$ has to be robust to the background signal $S$, such that the estimation $\tilde{\theta}$ mostly depends on the pattern $J$.
	\item The estimator $\tilde{\theta}$ in \eqref{eq:tildetheta} should be computable only based on the knowledge of the measurements $(q_\alpha$). 
\end{itemize}  
We now identify the filters $\xi_\alpha$ and $\xi$ that allow us to achieve these three goals. To do so, we rely on steerable filters. 
We briefly introduce the concepts of steerability and refer the reader to Appendix \ref{sec:steerability} for more details.

We consider filters $\xi_\alpha \in L_2(\mathbb{R}^2)$ that are polar-separable and have the form 
\begin{equation}
  \label{eq:polsep}
  \xi_\alpha(r,\theta) = \eta_\alpha(r) \ue^{\uj n_\alpha\theta}
\end{equation}
with $\eta_\alpha$ the \textit{radial profile} 
and $n_\alpha\in\ZZ$ the \emph{harmonic} of $\xi_{\alpha}$. 
We assume that $\xi_\alpha$ is normalized such that $\lVert \xi_\alpha \rVert_2 = 1$. The polar separability of the $\xi_\alpha$'s implies that
\begin{align} \label{eq:scalarxis}
	\langle \xi_\alpha , \xi_\beta \rangle &= 
	\int_0^{\infty}   \eta_{\alpha} (r) \eta_{\beta}(r) r \mathrm{d} r \int_0^{2\pi} \ue^{\uj (n_\alpha - n_\beta) \theta} \frac{\mathrm{d}\theta}{2\pi} \nonumber \\
	&= \delta[n_\alpha - n_\beta] \int_0^{\infty}   \eta_{\alpha} (r) \eta_{\beta}(r) r \mathrm{d} r
\end{align}
with $\delta[\cdot]$ the Kronecker delta.
This shows that the system $(\xi_\alpha)$ is orthonormal when the $n_\alpha$ are all distinct.

The measurement function $\xi_\alpha$ being polar-separable, its Fourier transform is also polar-separable as
\begin{align}  
\hat{\xi}_\alpha(\omega,\varphi) = \hat{h}_\alpha(\omega) \ue^{\uj n_\alpha\varphi}.
\label{eq:FourSepN}
\end{align}
The function $\hat{h}_\alpha$ is related to $\eta_\alpha$ and $n_\alpha$ by the Hankel transform (see, \emph{e.g.}, \cite[Proposition 2]{ZsSteerable}).

We remark that the rotated version of $\xi_{\alpha}$ satisfies 
\begin{equation}\label{eq:rotatexialpha}
\xi_{\alpha}(\M R_{\theta_0} \V x ) = \xi_{\alpha}( r , \theta + \theta_0) = \mathrm{e}^{\mathrm{j} n_\alpha \theta_0} \xi_{\alpha}(\V x),
\end{equation} 
and is therefore steerable in the sense of Definition \ref{def:steerable} in Appendix \ref{sec:steerability}. 
We then select a filter  of the form $\xi = \sum_{\alpha} c_{\alpha} \xi_\alpha$, where the vector $\V c = (c_\alpha)$ determines the shape of the filter. 
The way of selecting  $\V c$ adequately is discussed in Proposition \ref{prop:MLE}.
Hence, $\xi$ is the best possible filter for the pattern $J$ in the mean-square sense, once the measurement functions are given.

As a linear combination of steerable filters (see \eqref{eq:rotatexialpha}), the filter $\xi$ is steerable  with 
\begin{equation} 
	\xi( \M R_{\theta_0} \V x ) = \sum_\alpha c_\alpha \mathrm{e}^{\mathrm{j} n_\alpha \theta_0}\xi_\alpha(\V x).
\end{equation}
This means that it is sufficient to apply the filtering with $\xi_\alpha$ only once. Then, the rotated filter for any arbitrary angle can be determined by a systematic and linear transformation of the initial basis filters.


\subsection{The Law of Measurement Vector and the Cram\'er-Rao Lower Bound}
\label{sec:CRLB}

We aim at evaluating the performance of the estimators of $\theta^*$ depending on the measurement vectors $\V q = (q_{\alpha})$. In this section, we only consider estimators $\tilde{\theta} = \tilde{\theta}(\V q)$ that are unbiaised; that is,
\begin{equation} \label{eq:unbiaisedestimator}
	\mathbb{E} [ \tilde{\theta}(\V q )] = \theta^*.
\end{equation}
The performance of the estimator is then measured by the mean-square error
\begin{equation}\label{eq:meansquareerror}
	\mathbb{E} [ (\tilde{\theta}(\V q ) -\theta^*)^2].
\end{equation}

As is well-known, there is a theoretical bound for the mean-square error that cannot be surpassed by an unbiaised estimator, called the Cramer-Rao lower bound (CRLB). The latter is given by (see \cite{kay1993fundamentals}) 
\begin{equation} \label{eq:CRLBdef}
\mathbb{E} [ (\tilde{\theta}(\V q ) -\theta^*)^2] \geq 1 /  \mathrm{FI}(\theta^*) = \mathrm{CRLB}(\theta^*),
\end{equation}
where $\mathrm{FI}(\theta^*)$ is the Fisher information of the measurement vector $\V q$.

In Theorem \ref{prop:Statsq}, we provide the mean vector, the covariance matrix and the Gaussian measurement vector and deduce its Fisher information. We recall here that the measurement filters $\xi_\alpha$ are polar-separable of the form \eqref{eq:polsep}. 

\begin{theorem} \label{prop:Statsq}
Consider measurement filters of the form \eqref{eq:polsep}. We assume that the $\xi_\alpha$ are such that 
\begin{equation} \label{eq:conditionxialpha}
\int_{0}^{\infty} \omega^{1-2\gamma} \lvert \hat{h}_{\alpha} (\omega) \rvert^2 \mathrm{d}\omega < \infty,
\end{equation} 
where the $\hat{h}_\alpha$ are defined in \eqref{eq:FourSepN}.
Then, the measurement vector $\V q$ is Gaussian with mean vector
\begin{equation} \label{eq:meanq}
	\V \mu = \V \mu (\theta^*) = \left( \mathrm{e}^{\mathrm{j} n_\alpha \theta^*} \langle J , \xi_\alpha\rangle \right)_\alpha
\end{equation}
and covariance matrix
\begin{equation} \label{eq:covq}
	\M C = \left(  \delta[n_\alpha - n_\beta] \int_0^{\infty} \omega^{1 - 2 \gamma} \hat{h}_\alpha(\omega) 	\overline{\hat{h}_\beta (\omega)} \mathrm{d}\omega \right)_{\alpha,\beta},
\end{equation}
where the $h_\alpha$ are given in \eqref{eq:FourSepN} and $\gamma$ is the order of the whitening operator of $S$.
The Fisher information is then
\begin{equation}\label{eq:FIq}	
	\mathrm{FI}(\theta^*) 
	= 2 \mathrm{Re} 
	\left( 
	\left( \frac{\mathrm{d} \V \mu}{\mathrm{d} \theta^*}\right)^H 
	\M C^{-1} 
	\frac{\mathrm{d} \V \mu}{\mathrm{d} \theta^*} 	
	\right),
\end{equation}
where $\V x^H = (\V x^*)^T$ is the Hermitian transpose of $\V x$.
\end{theorem}

The proof of Theorem \ref{prop:Statsq} is postponed to Appendix \ref{sec:proofGaussianvector}. The condition \eqref{eq:conditionxialpha} is necessary to make the problem well-posed, because $\langle S , \xi_\alpha \rangle$ is well-defined under this condition. It is related to the existence of vanishing moments for $\xi_\alpha$ (see Appendix \ref{sec:noise}).
 It is important to notice that the mean $\V \mu (\theta^*)$ of $\V q$ depends on $\theta^*$, while the covariance $\M C$ does not. This is due to the assumption that the background signal is statistically isotropic.

\subsection{Maximum Likelihood Estimator}
\label{sec:MLE}

In this paper, we consider estimators of the form \eqref{eq:tildetheta}, where $\xi$ is a steerable filter given by $\xi = \sum_\alpha c_\alpha \xi_\alpha$. We show here that, under some assumptions, this corresponds to the maximum likelihood estimator for an adequate choice of the $c_\alpha$. More precisely, one selects the $c_\alpha$ such that $\xi$ corresponds to the orthogonal projection of the pattern $J$ to the basis filters ($\xi_\alpha$).

\begin{proposition} \label{prop:MLE}
	We assume that the filters $\xi_\alpha$ have distinct harmonics ($n_\alpha \neq n_\beta$ for $\alpha \neq \beta$) and identical radial profile ($\eta_\alpha = \eta$ for every $\alpha$). Then, the maximum likelihood estimator 
	\begin{equation}
		\theta_{\mathrm{MLE}} = \mathop{\text{argmax}}_{\theta_0 \in [0,2\pi)} \mathbb{P} ( \V q | \theta_0 ) 
	\end{equation}
	of the image model \eqref{eq:imagemodel} coincides with \eqref{eq:tildetheta} with
	\begin{equation}
		\xi = \sum_{\alpha} \langle J , \xi_\alpha \rangle \xi_\alpha.
	\end{equation}
\end{proposition}

\begin{proof}
	According to Theorem \ref{prop:Statsq}, conditionally to $\theta^* = \theta_0$, the vector $\V q$ is Gaussian with mean $\V \mu(\theta_0)$ and co-variance matrix $\M C$. Therefore,  $ \mathbb{P} ( \V q | \theta_0 )$ is proportional to $\exp( - \frac{1}{2} (\V q - \V \mu (\theta_0) )^H \M C^{-1} (\V q - \V \mu (\theta_0) )$. Hence,
	\begin{align}
		\theta_{\mathrm{MLE}} &= \mathop{\text{argmin}}_{\theta_0 \in [0,2\pi)} (\V q - \V \mu (\theta_0) )^H \M C^{-1} (\V q - \V \mu (\theta_0) ) .
	\end{align}
	The harmonics of the $\xi_\alpha$ being distinct, $\M C$ is diagonal. Moreover, $\M C[n_\alpha,n_\alpha]$ depends only on $\eta_\alpha = \eta$, and is therefore independent of $\alpha$. This means that 
		\begin{align}
		\theta_{\mathrm{MLE}} &= \mathop{\text{argmin}}_{\theta_0 \in [0,2\pi)} (\V q - \V \mu (\theta_0) )^H \M  (\V q - \V \mu (\theta_0) ) .
	\end{align}
	One can develop this latter expression. We remark that $\V q^H \V q$ and $\V \mu(\theta_0)^H \V \mu (\theta_0) = \sum_{\alpha} \langle J, \xi_\alpha\rangle^2$ are independent of $\theta_0$. Therefore, we have finally,
	 		\begin{align}
		\theta_{\mathrm{MLE}} &= \mathop{\text{argmin}}_{\theta_0 \in [0,2\pi)}  (-2 \V q^H \V \mu (\theta_0)) \nonumber \\
		&= \mathop{\text{argmax}}_{\theta_0 \in [0,2\pi)} \sum_\alpha q_\alpha \langle J , \xi_\alpha\rangle \mathrm{e}^{\mathrm{j} n_\alpha \theta_0}.
	\end{align}
	Moreover, recalling  that $q_\alpha = \langle I , \xi_\alpha \rangle$, we can rewrite \eqref{eq:tildetheta} as 
\begin{align}
\tilde{\theta} &= \mathop{\text{argmax}}_{\theta_0 \in [0,2\pi)}{
\left\langle I, \sum_{\alpha} c_\alpha \ue^{\uj n_\alpha  \theta_0} \xi_\alpha \right\rangle} \notag\\
&= \mathop{\text{argmax}}_{\theta_0 \in [0,2\pi)}{ \sum_{\alpha} q_\alpha {c_\alpha} \ue^{\uj n_\alpha \theta_0} },
\label{estimatorEq}
\end{align}
Hence, $\tilde{\theta} = \theta_{\mathrm{MLE}}$ as soon as $c_\alpha = \langle J , \xi_\alpha \rangle$, as expected. 
\end{proof}

The angle $\tilde{\theta}$ in \eqref{estimatorEq} is numerically computed with standard optimization techniques or, simply, with a grid search.
Finally, $\tilde{\theta}$ is an estimator of the form \eqref{eq:tildetheta}, which depends only on the measurements $q_\alpha$, and is based on the steerable filter $\xi$ that provides the best approximation of the pattern $J$ in span $\left\{\xi_\alpha\right\}$.

We shall now use the CRLB \eqref{eq:CRLBdef} together with the Fisher information \eqref{eq:FIq} to evaluate the optimal performance achievable by an unbiaised estimator.
We shall give there empirical evidence that the proposed estimator $\tilde{\theta} = \theta_{\mathrm{MLE}}$ is unbiaised and show that it empirically achieves the CRLB in simulations.

\section{Cram\'er-Rao Lower Bound for an Estimation with Distinct Harmonics}
\label{sec:crlb1}

\subsection{The Cram\'er-Rao Lower Bound}
\label{finite}

In this section, we derive the CRLB in cases where the measurement functions defined in \eqref{eq:polsep} have distinct harmonics $n_\alpha$. So that $n_\alpha =n$, we shall consequently index the measurement functions directly by their harmonics $n\in H$, where $H\subset\ZZ$ is the set of all used harmonics. Specifically, we write
\begin{equation} \label{eq:coefficients}
	\xi_n(r,\theta) 
	= 
	\eta_n(r) \ue^{\uj n \theta}
\end{equation}
and, in the Fourier domain,
\begin{equation}\label{eq:xinF}
  \hat{\xi}_n(\omega,\varphi) 
  = 
  \hat{h}_n(\omega) \ue^{\uj n \varphi}.
\end{equation}

We recall that for $n\neq m$, $\xi_n$ and $\xi_m$ are orthogonal, due to the orthogonality of their angular factors (see \eqref{eq:scalarxis}).
%
Moreover, we have that, for a real image $I$,
\begin{equation}
	q_{-n} 
	= 
	\langle I, \xi_{-n}\rangle 
	= 
	\langle I, \overline{\xi_n} \rangle 
	= 
	\overline{\langle I,\xi_n\rangle} 
	= 
	\overline{q_n},
\end{equation}
where we used the fact that $\eta_n$ is real. 
Thus, $q_{-n}$ and $q_n$ essentially carry the same information, so that the CRLB based on $q_n$, $n\in H$, is the same as the CRLB based on $q_n$ with harmonics $n$ in the set
\begin{equation}
	H_+ 
	= 
	\left\{|n| : n\in H\right\}.
\end{equation}
We further exclude $n=0$ from consideration since the corresponding measurement does not depend on the rotation angle $\theta^*$.
We remark that the correlation matrix $\M C$ in \eqref{eq:covq} is diagonal due to the distinct nature of the harmonics. As a consequence, the CRLB does not depend on $\theta^*$, as is easily deduced from Proposition \ref{prop:Statsq}. 

We now calculate the CRLB for estimating the angle $\theta^*$ of the pattern in \eqref{eq:imagemodel}.
\begin{theorem}
\label{Th:distinctH}
For measurements using distinct harmonics $n\in H$, 
the exact form of the CRLB is
\begin{equation}
	\mathrm{CRLB} 
	=  
	\frac{\sigma_0^2 / 4\pi}{\sum\limits_{n \in  H_+} 
	\frac{n^2\abs{u_{n}}^2}{\int_0^{\infty} 
	\omega^{1 - 2\gamma} 
	\left\vert \hat{h}_n(\omega)\right\vert ^2 \ud\omega}},
\label{OneScaleCRLB}
\end{equation}
where we recall that $u_n = \langle J , \xi_n\rangle$. 
\end{theorem}
The proof of Theorem \ref{Th:distinctH} is given in Appendix \ref{sec:distinctHproof}.

\subsection{Conventional Detector: Estimation From Best $N$ Measurements}
\label{bestN}

In the case of conventional detectors, one chooses the same radial pattern $\hat{h}$ for all measurement functions \cite{NicolasSIAM}, which results in
\begin{equation}
  	\hat{\xi}_n(\omega,\varphi) 
  	= 
  	\hat{h}(\omega)\ue^{\uj n\varphi}.
\end{equation}
This radial component is typically band-pass, given its vanishing moments and finite energy. In practice, the choice of the radial profile (bandpass filter $\hat{h}$) specifies the scale of the detector. 

Now we address the question of selecting the best steerable subspace for specifying our steerable matched-filter detector.
Suppose we have at our disposal the finite number $N$ of harmonics (\textit{i.e.}, a finite number of measurements). It is then natural to ask which harmonics to consider to reduce the CRLB as much as possible. This is obtained after an immediate corollary of Theorem \ref{Th:distinctH}.

\begin{corollary}\label{Th:allHfixedS}
The CRLB for the estimation problem in the case of a single common radial profile is
\begin{equation}\label{eq:CRLBsinglescale}
	\mathrm{CRLB} 
	=  
	\frac{1}{\sum\limits_{n\in H_+} n^2 
	\abs{u_{n}}^2} 
	\frac{\sigma_0^2}{4\pi} 
	\int_0^{\infty} 
	\omega ^{1 - 2\gamma} 
	\left\vert \hat{h}(\omega)\right\vert ^2 \ud\omega.
\end{equation}
\end{corollary}
Corollary \ref{Th:allHfixedS} simply exploits the fact that the $\hat{h}_n$ are all equal to $\hat{h}$ in \eqref{OneScaleCRLB}.

Considering \eqref{eq:CRLBsinglescale}, choosing the best measurements is equivalent to identifying the set $H_+\subset\ZZ_+$ with $N$ members, for which the sum
\begin{equation}\label{summaximized}
  	\sum_{n\in H_+} 
  	n^2|u_n|^2
\end{equation}
is maximized.

\subsection{Conventional Detector: Asymptotic Behavior}
\label{asymptotic}

Similarly to Section \ref{bestN}, we fix the radial part of $\xi_n$  as being independent of $n$ and look at the asymptotic behavior in terms of the number of harmonics, by choosing the set $H_+$ of harmonics as $\{1,\ldots,N\}$, and letting $N$ tend to infinity.
Based on (\ref{eq:CRLBsinglescale}), we conclude that the asymptotic behavior of the CRLB depends on the asymptotic (decay) properties of the coefficients 
$u_n = \langle J,\xi_n\rangle$ of the directional pattern.

The question is whether the CRLB vanishes asymptotically as $N\to\infty$, which would suggest the theoretical possibility of perfect estimation with infinitely many measurements. 
We now study this question.

We start with a preliminary remark. By assumption, the pattern $J$ has finite energy. Since the measurement functions $\{ \xi_n \}_{n\in\ZZ}$ are orthonormal, by the Bessel inequality, we have that
\begin{equation}
	\sum_{n\in\ZZ} 
	\abs{u_{n}}^2 
	\leq 
	\norm{J}_2^2 
	< 
	\infty,
\end{equation}
hence $(u_{n})_{n\in \ZZ} \in \ell_2 (\ZZ)$. We see now how to refine this latter condition to have a vanishing CRLB.

\begin{proposition} \label{prop:asymptoticsinglescale}
	In the framework of conventional detectors, for a pattern $J \in L_2(\RR^2)$, the CRLB does not vanish when $N\rightarrow \infty$ if and only if
	\begin{equation} \label{eq:sumfiniteCRLB}
		\sum_{n \in \ZZ} n^2 \lvert u_n \rvert_2^2 < \infty.
	\end{equation}
	Moreover, the $u_n$ are the Fourier coefficients of the function 
	\begin{equation} \label{eq:defG}
  	G (\varphi) 
  	= 
  	\frac{1}{2\pi}
  	\int_0^\infty 
  	\overline{  \hat{J} (\omega,\varphi)}
  	\hat{h}(\omega)\omega
  	\ud\omega.
\end{equation}
Therefore, the CRLB does not vanish  if and only if the function $G$ is differentiable with its derivative being square integrable. 
\end{proposition}

\begin{proof}
	The CRLB is inversely proportional to \eqref{eq:sumfiniteCRLB}, implying the first equivalence. We then remark that, using Parseval relation and polar coordinates, we have
	\begin{align}\label{eq:unGn}
	u_n 
	&= 
	\langle J,\xi_n\rangle 
	= 
	\frac{1}{(2\pi)^2}\langle  
	\hat{J} ,\hat{\xi}_n\rangle \notag\\
  	&= 
  	\frac{1}{(2\pi)^2}
  	\int_0^{2\pi} 
  	\int_0^\infty \overline{  \hat{J} (\omega,\varphi)}\hat{h}(\omega)
  	\ue^{\uj n\varphi}\omega
  	\ud\omega
  	\ud\varphi 
  	\notag\\
  	&= 
  	\frac{1}{2\pi}
  	\int_0^{2\pi}  
  	G (\varphi)
  	\ue^{\uj n\varphi}
  	\ud\varphi.
\end{align}
The sequence $(u_n)$ is in $\ell_2(\ZZ)$, therefore $G$ is square integrable. The Fourier coefficients of the (weak) derivative of $G$ are then $\mathrm{j} n u_n$.
Hence, $G$ is differentiable with a square integrable derivative if and only if the sequence $(\mathrm{j} n u_n)$ is in $\ell_2(\ZZ)$, proving the second equivalence in Proposition \ref{prop:asymptoticsinglescale}. 
\end{proof}

The convergence of the series in \eqref{eq:CRLBsinglescale} therefore depends on the decay properties of the Fourier-series of $ G (\varphi)$, which in turn is related to its smoothness.
In particular, if $ \hat{J}$ has angular jump discontinuities that are inherited by $ G $, $u_n$ will decay slowly like $1/n$, and the series will diverge. In this case, the CRLB will asymptotically vanish. This, for instance, happens if $ \hat{J} $ has jump discontinuities along infinite radial lines, which typically goes along with a similar discontinuity in $J$ (see \cite{Smith1973starfunction}). This suggests that, for such patterns, the angular error of a steerable detector can be made arbitrarily small by selecting a sufficient number of harmonics.

\section{Experiments}
\label{experiments}

In practice, the radial profil $\hat{h}$ is typically chosen as the Laplacian of a Gaussian (LoG) or a band-pass filter (even and compactly supported in the frequency domain). 
For the experiments, we use the LoG filter and the first scale of the Meyer-type profile \cite{daubechies1992ten}

\begin{align}
\hat{h}(\omega) 
= 
\begin{cases}
\sin
\left(\frac{\pi}{2}\nu\left(\frac{4\omega}{\pi}-1\right)\right), 
& \frac{\pi}{4} < \omega \leq \frac{\pi}{2}\\
\cos
\left(\frac{\pi}{2}\nu\left(\frac{2\omega}{\pi}-1\right)\right), 
& \frac{\pi}{2} < \omega \leq \pi\\
0, 
& \text{otherwise}\\
\end{cases},
\label{eq:Meyer}
\end{align}
with the auxiliary function $\nu(t) = t^4(35 - 84t + 70t^2 - 20t^3)$, $\nu \in C^3([0,1])$.

Other typical examples include Shannon-type \cite{DNU}, Simoncelli \cite{PortillaParametric}, Papadakis \cite{Papadakis} and Held \cite{held2010steerable} wavelets. 
Similar results are obtained when using these radial profiles  due to the fact that all these functions approximate the indicator function $\left[ \pi/4, \pi/2 \right]$. 

\subsection{CRLB for Analytical Patterns}
\label{CRLBAnEx}

In this section, we compute the CRLB associated with a few directional patterns for which explicit formulas are provided. We study four different types of patterns. Specifically,
\begin{align}
\hat{J}_{1}(\omega,\varphi) &=
\begin{cases}
1, & \text{if} \quad \cos(1.5 \varphi)^{\beta} > 0.8 \\
0, & \text{otherwise}
\end{cases}
\label{J1} \\
\hat{J}_{2}(\omega,\varphi) &=  \left(\frac{1}{1+\omega^\lambda}\right)\cos(1.5 \varphi)^{\beta}
\label{J2} \\
\hat{J}_{3}(\omega,\varphi) &=
\begin{cases}
1, & \text{if} \quad \cos(2 \varphi)^{\beta} > 0.8 \\
0, & \text{otherwise}
\end{cases}
\label{J3}
\\
\hat{J}_{4}(\omega,\varphi) &=  \cos(2 \varphi)^{\beta} \ue^{-\sfrac{1}{\alpha \omega}}
\label{J4}
\end{align}
within the support of $\hat{h}(\omega)$, with $\lambda = 2.1$, $\beta = 28$ and $\alpha = 2.5$.

We are interested in the following quantities: $\vert u_n \vert$, $n \vert u_n \vert$ (as it determines the decay rate of the CRLB), and the CRLB for a fixed number of harmonics. For computing the CRLB, we apply three different strategies: ``First $N$''; ``Best $N$''; and ``$k$-fold''. In the first case, we use the first $N$ coefficients; in the Best $N$ case, we select the harmonics that maximize (\ref{summaximized}). In the case of $k$-fold symmetric patterns we choose the first $N$ multiples of $k$ as harmonics. This latter choice accounts for the name ``$k$-fold''.

Related to (\ref{OneScaleCRLB}), we have chosen the value $\sigma_0 =1$ since it provides only a scaling factor and does not influence the decay of the curve.
Figure \ref{im:analytic3} contains an illustration of the results.

\begin{figure}[!t]
\centering
\subfloat{\includegraphics[width = 0.48\textwidth]{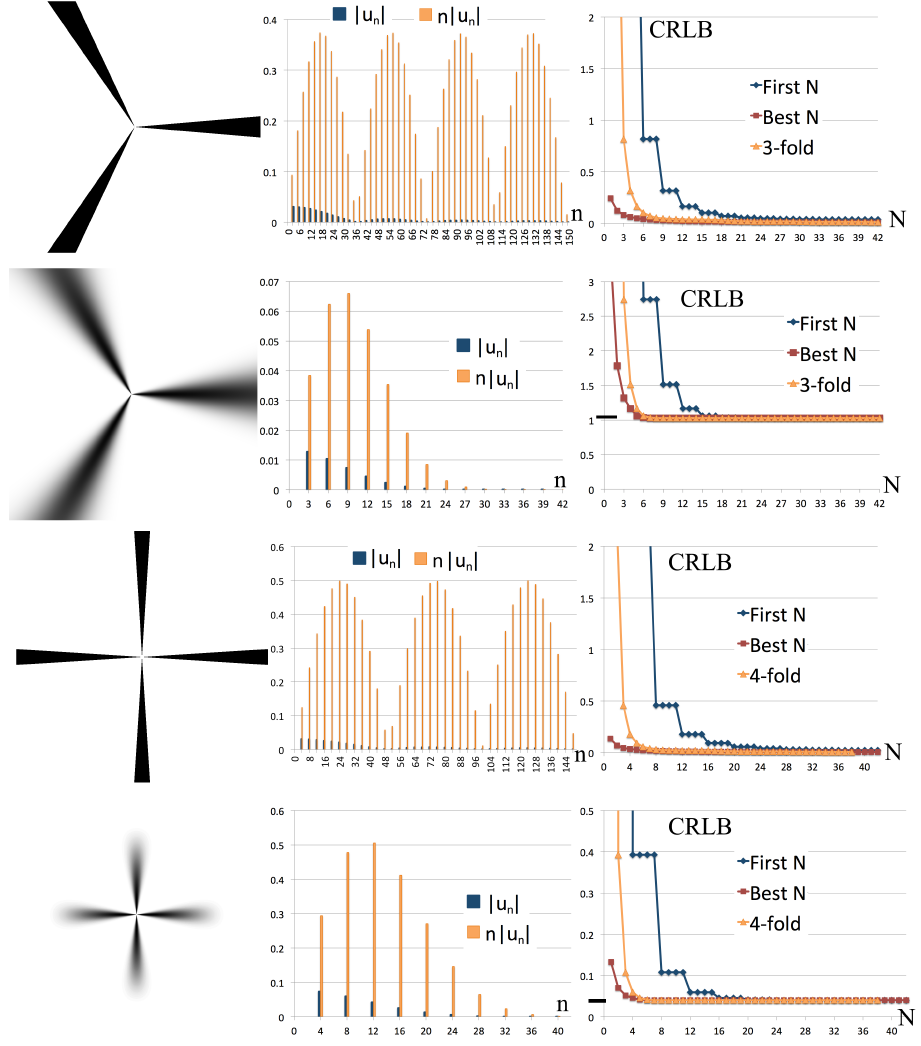}}
\caption{
First column from top to bottom: 
illustration of the analytically defined patterns $\hat{J}_1$ (\ref{J1}), $\hat{J}_2$ (\ref{J2}), 
$\hat{J}_3$ (\ref{J3}) and $\hat{J}_4$ (\ref{J4}) in the Fourier domain. 
Second column: 
$\vert u_n \vert$ and $n \vert u_n \vert$ as a function of harmonics ($n$). 
Third column: 
The CRLB as a function of the number of harmonics.}
\label{im:analytic3}
\end{figure}

In the case of sharp edges ($\hat{J}_1$ (\ref{J1}) and $\hat{J}_3$ (\ref{J3})), the rate of decay of the circular-harmonic coefficients permits a theoretical vanishing limit for the CRLB.
In the smooth angular cases ($\hat{J}_2$ (\ref{J2}) and $\hat{J}_4$ (\ref{J4})), the CRLB converges to a theoretical positive value. 

Moreover, as expected, for the three-fold patterns, only every third component and for the four-fold patterns every fourth element plays a significant role in the estimation of the orientation. This can be seen in the almost flat CRLB curve between multiples of three (or four, respectively) in the ``First $N$'' strategy. We also observe a difference in performance between generically choosing the first $N$ $k$-fold symmetric coefficients (as a strategy for unknown $k$-fold patterns) and making our choice of harmonics based on maximizing (\ref{summaximized}).

Finally, by looking at Figure \ref{im:analytic3}, we observe that, with the right choice of harmonics, the CRLB can be much reduced, even with a small number of harmonics.

\begin{figure*}[!t]
\centering
\subfloat{
\includegraphics[width=.95\textwidth]{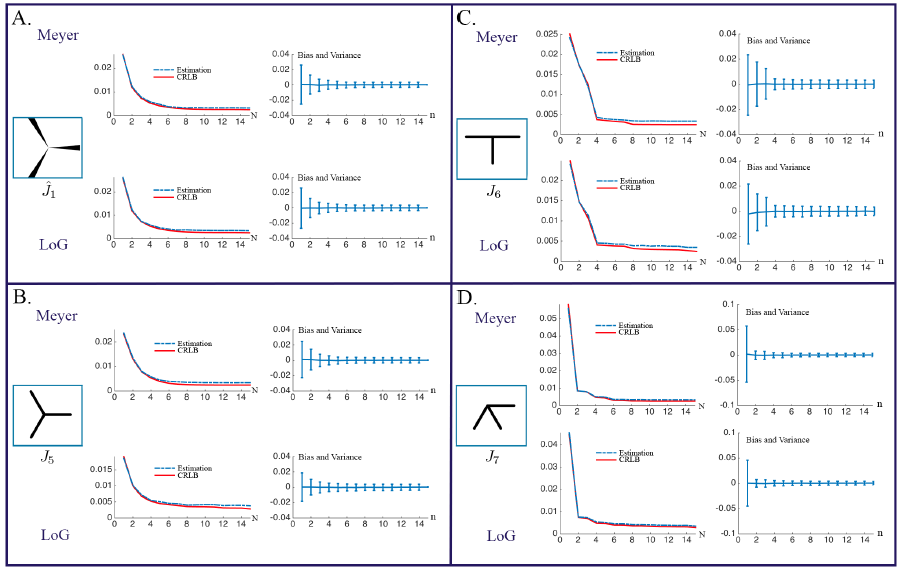}
}
\caption{Accuracy of the orientation estimation of the proposed estimator (dashed line) compared to the CRLB (continuous line), bias and variance of the proposed estimator. 
A.: Results on the analytically defined junction $\hat{J}_1$ (\ref{J1}); top: with Meyer radial detector profile, bottom: with LoG detector profile. B. Results on a three-fold symmetric junction (drawn), top: with Meyer radial detector profile, bottom: with LoG detector profile. C. Results on a T-shaped junction (drawn); top: with Meyer radial detector profile, bottom: with LoG detector profile. D. Results on an arrow-shaped junction (drawn); top: with Meyer radial detector profile, bottom: with LoG detector profile.}
	\label{im:estimator}
\end{figure*}

\subsection{Empirical Optimality of the Proposed Estimator}

In this section, we test the accuracy of the estimation of the proposed estimator (\ref{estimatorEq}) on two symmetric and two asymmetric patterns. For symmetric patterns, we have chosen the analytically defined junction $J_1$ (\ref{J1}) and a sharp, drawn three-fold junction (Figure \ref{im:estimator}, $J_5$). For asymmetric patterns, we have chosen a sharp, drawn T-shape junction (Figure \ref{im:estimator}, $J_6$) and a sharp, drawn arrow-type junction (Figure \ref{im:estimator}, $J_7$). The coefficients $c_{\alpha}$ in \eqref{estimatorEq} are obtained in each case as the orthogonal projection of the junction of interest on the measurement functions. 

We chose the variance of the noise such that it corresponds to an SNR of 17.22 dB. The typical SNR range of real images where the quality is still acceptable is 15-35 dB, so that our experimental conditions are representative of a noisy image. We built $1,\!000$ different realizations to make the experiments statistically reliable.

For the estimator, we used every third harmonic. For the radial part of our detector, we have chosen the first scale of the Meyer wavelet profile and the Laplacian of Gaussian (LoG) filter.

The results for the patterns are illustrated in Figure \ref{im:estimator}. 
We first show experimentally that the estimator \eqref{estimatorEq} is unbiased with more than one harmonic. 
The performances are therefore comparable to the CRLB. 
We then observe that the accuracy of the estimator follows closely the CRLB curve, while staying above, as expected.
This proves empirically the almost optimality of the estimator \eqref{estimatorEq} for the estimation of the angle $\theta^*$.


\section{Extension to Wavelets}
\label{sec:crlb2}


\subsection{Steerable Wavelets}

In this section, we extend our results on the CRLB to wavelet detectors. 
Steerable wavelet frames are adapted to capture the local orientation of features, or junctions, within a multiresolution hierarchy. To simplify the notations, we consider wavelets that are centered at the origin.
Moreover, we apply the multiorder complex Riesz transform on a tight wavelet frame of $L_2 \left(\mathbb{R}^2 \right)$.
Proposition 4.1 in \cite{NicolasSIAM} provides sufficient conditions on the isotropic profile that has to be defined to generate a desired wavelet system.
There are various types of isotropic profiles satisfying the proposition; typical examples are given at the beginning of Section \ref{experiments}. For further details, we refer to \cite{DNU}.

Similarly to the design of conventional detectors, we take the values of $n$ from the predefined set of harmonics $H =\{n_{0},\ldots,n_N\}$. The wavelet schemes generated in such a way are often referred to as circular harmonic wavelets \cite{NicolasSIAM,Neri}. 
In an extension of \eqref{eq:coefficients}, the measurement functions are indexed here by the pair $\alpha=(n,i)$ of the harmonic $n\in H$ and the scale $i\in \ZZ$.
In particular, the circular harmonic wavelet $\xi$ at scale $i$ and harmonic channel $n$ takes the form
\begin{align}
\hat{\xi}_{n,i} (\omega,\varphi) = 2^i \hat{h} \left(2^i\omega \right) \ue^{\uj n\varphi}.
\end{align}
We note that, in the case of wavelet measurements, we have multiple measurement functions (at different scales) for the same harmonic $n$. 
We make the additional assumption that $\hat{h}$ is real-valued, which holds true for every aforementioned radial design. Finally, given that the radial patterns $\hat{h}$ of interest are bandpass, we assume that $\hat{h}(\omega) = 0$ for $\omega \notin (\pi/4, \pi ]$ to set the fundamental scale.

\subsection{Multiple Scales: The Cram\'er-Rao Lower Bound}

We recall that the goal is to estimate $\theta^*$ from the measurements $q_{n,i}=\langle I,\xi_{n,i}\rangle$. 
We also denote  $u_{n,i} = \langle J,\xi_{n,i}\rangle$ and $s_{n,i}=\langle S,\xi_{n,i}\rangle$.

To simplify future formulas, we introduce the following notations:
\begin{align}
  \tilde{q}_{n,i} & = 2^{-i\gamma} q_{n,i},\\
  b_z & =  \frac{1}{2\pi}  \int_0^{\infty} \omega^{z} \hat{h}(\omega)^2 \omega \ud\omega,  \label{eq:bzzz}\\
  d_z & =   \frac{1}{2\pi} \int_0^{\infty}  \omega^{z} \hat{h}\left(\omega\right) \hat{h}\left(2\omega\right) \omega \ud\omega.
\end{align}
We also set the two constants
\begin{equation}
\label{eq:BD}
  B = \sigma_0^2 b_{-2\gamma}, \quad
  D = \sigma_0^2  2^{1-\gamma} d_{-2\gamma}.
\end{equation}

The angle ${\theta^*}$  is estimated by steering the whole template, which can be seen as a sum of templates. The task of estimating ${\theta^*}$ based on ${q_{n,i}}$ is the same as its estimation based on  ${\tilde{q}_{n,i}}$. In particular, the CRLB is the same. We first give the covariance matrix of the measurements $\tilde{q}_{n,i}$.

\begin{proposition} \label{prop:covwavelets}
The covariance matrix $\M C$ of the random vector $(\tilde{q}_{n,i})_{n\in \mathbb{Z}, i \in \mathbb{Z}}$ is given by

\begin{equation}
  \M C [(n,i), (m,k)]
  =
  \begin{cases}
    B, & \text{if} \quad m=n, k=i \\
    D, & \text{if} \quad m=n, \left\vert k-i\right\vert =1 \\
    0, & \text{otherwise.}
  \end{cases}
  \label{eq:waveletcov}
\end{equation}

\end{proposition}

\begin{proof}
First, $\M C [(n,i), (m,k)] = 0$ as soon as $m\neq n$ because the functions $\xi_{n,i}$ and $\xi_{m,k}$ are orthogonal. Then, the assumption that  $\widehat{h} (\omega) = 0$ for $\omega \notin (\pi/4, \pi]$ ensures that $\widehat{h}(2^i \omega)$ and $\widehat{h}(2^k \omega)$ do not overlap for $|k-i|>1$, implying again that $\M C [(n,i), (m,k)] = 0$. 

We now assume that $m=n$ and $|k-i| \leq 1$. When $k=i$, we easily recognize the quantity \eqref{eq:bzzz} with $z=-2\gamma$  in \eqref{eq:covq}  and therefore deduce that $\M C [(n,i), (n,i)] = \sigma^2_0 b_{-2\gamma} = B$. Similarly, one shows using again \eqref{eq:covq} that $\M C [(n,i), (n,k)] = D$ when $|k-i| = 1$.
\end{proof}

The main challenge in computing the Fisher information from \eqref{eq:FIq} for wavelet measurements is that, unlike in context of Section \ref{sec:crlb1}, the covariance matrix of the measurements $\{\tilde{q}_{n,i}\}$ is not diagonal, so it is not as straightforward to invert. But this challenge is still surmountable because, as we see from \eqref{eq:waveletcov}, we can rearrange the measurements such that the covariance matrix is (at most) tridiagonal. The exact rearrangement of the measurements is described in details in Appendix \ref{sec:covmat}.

The key idea is to divide the set of measurement indices into as few disjoint sets $G_r$ as possible ($r=1,\ldots,g$, where $g$ denotes the total number of such sets). Then, each set $G_r$ has some $l_r$ elements of the form $(n_r,i_r),\ldots,(n_r,i_r+l_r-1)$, for some harmonic $n_r\in H_+$ and minimum scale $i_r$. We re-index all measurements as $\tilde{q}_{r|e}$ using the notation 
\begin{equation}
(r|e) := (n_r,i_r+e-1),
\label{eq:rle}
\end{equation}
with $r=1,\ldots,g$ and $e=1,\ldots,l_r$.
Then, the covariance matrix of the measurements $\{\tilde{q}_{r|e}\}$ for each fixed $r$ is Toepliz-tridiagonal of the form
\begin{equation}
  \M T_r =
  \begin{pmatrix}
    B & D  & 0 & 0& \dots  & 0 \\
    D & B  & D  & 0&  \dots & 0 \\
    0 & D & B  & D  & \dots  & 0\\
    \vdots & \vdots & \vdots  & \ddots & \vdots & \vdots \\
        0 & 0 &  \dots & D & B & D \\
    0 & 0 &  \dots & 0 & D & B
  \end{pmatrix}_{l_r\times l_r}.
  \label{eq:Tr}
\end{equation}
Moreover, the overall covariance matrix is block-diagonal, with $\M T_r$, $r=1,\ldots,g$, as its diagonal blocks.
Using this reformulation of the problem, we are able to compute the CRLB explicitly.

\begin{theorem}
  \label{Th:WaveletCRLB}
  For wavelet measurements ${q}_{r|e}$, $r=1,\ldots,g$, $e=1,\ldots,l_r$, the Fisher information is given by
  \begin{align}
     \text{FI}({\theta^*})
     &= 2\sum_{r=1}^g n_r^2 \sum_{t=1}^{l_r} \left( B +2D\cos\left(\frac{t\pi}{l_r+1}\right) \right)^{-1} \notag\\
   & {}\times{} \abs{ \sum_{e=1}^{l_r} 2^{-(i_r+e-1)\gamma} u_{r |  e} \sin \left(e\frac{t\pi}{l_r+1} \right) }^2.
  \end{align}
  The CRLB for the estimation problem is given by $1/\text{FI}(\theta^*)$. It satisfies
  \begin{align}
    \frac{B -  2 \abs{D}}{2\sum_{n,i} n^2 4^{-i\gamma} \abs{u_{n,i}}^2} \leq \mathrm{CRLB} \leq \frac{B+2\left\vert D\right\vert }{2\sum_{n,i} n^2 4^{-i\gamma} \left\vert u_{n,i}\right\vert ^2}.
    \label{CRLBwavelet}
  \end{align}
\end{theorem}

The proof of Theorem \ref{Th:WaveletCRLB} is given in Appendix \ref{proof:th2}.
The Fisher information is expressed as sums and inverses that are difficult to handle. However, we can give a fair estimate of the CRLB by narrowing it to a range \eqref{CRLBwavelet}. In the case of a conventional detector, the value of D is 0 (see \eqref{eq:BD}). Thus, the lower and upper bound of the range are equal and we arrive to the same exact CRLB that we computed for the case of conventional detectors \eqref{eq:CRLBsinglescale}. In the general case, it is crucial to remark that $0<B- 2|D|$, as demonstrated in Lemma \ref{lemma:ab} (see Appendix \ref{proof:th2}); hence the lower bound in \eqref{CRLBwavelet} is strictly positive. Finally, we remark that, as in the case of conventional detectors, the Fisher information, and therefore the CRLB, does not depend on $\theta^*$. 


\subsection{Asymptotic Behavior for All Scales and Harmonics}

We now consider the case where we have access to each wavelet coefficient $q_{n,i}=\left\langle I, \xi_{n,i} \right\rangle$ for $n\in\ZZ$, $i\in\ZZ$.
Once again, for real-valued data, we can limit ourselves to $n\in\ZZ_+$ for computing the CRLB, due to Hermitian symmetry.

\begin{proposition}
	\label{prop:asymptoticwavelet}
	Assume $\gamma$ is such that $\frac{b_0}{\abs{d_0}} \neq 2^{1-\gamma} + 2^{1+\gamma}$.
	Then,   by observing $\left\langle I, \xi_{n,i} \right\rangle$, the CRLB does not vanish if and only if 
	\begin{equation}
	 \sum_{i\in \ZZ} 	\sum_{n \in \ZZ^+} n^2 4^{i \gamma} |u_{n,i}|^2 < \infty.
	\end{equation}
	Moreover, if $J$ has an expansion in terms of $\{\xi_{n,i}\}$, then this is also equivalent to
	\begin{align}
(-\Delta)^{\gamma/2} \frac{\partial}{\partial \theta} J(r,\theta) \in L_2(\mathbb{R}^2).
	\end{align}
\end{proposition}
	
	The proof of Proposition \ref{prop:asymptoticwavelet} is provided in Appendix \ref{pr:gen}. 
	It is the multiscale version of Proposition \ref{prop:asymptoticsinglescale}.
	We interpret the fact that the CRLB vanishes as a possibility to perfectly detect the correct angle from the set of measurements, assuming all the harmonics and scales are available.
	The assumption on $\gamma$ is technical and simply implies that one specific value of $\gamma$ should be avoided for a given wavelet. The values of such $\gamma$ for different wavelets are given in Appendix \ref{pr:gen}.
		
\subsection{Experiments with Wavelets}

\begin{figure}[!t]
\centering
\subfloat{
	\includegraphics[width=.48\textwidth]{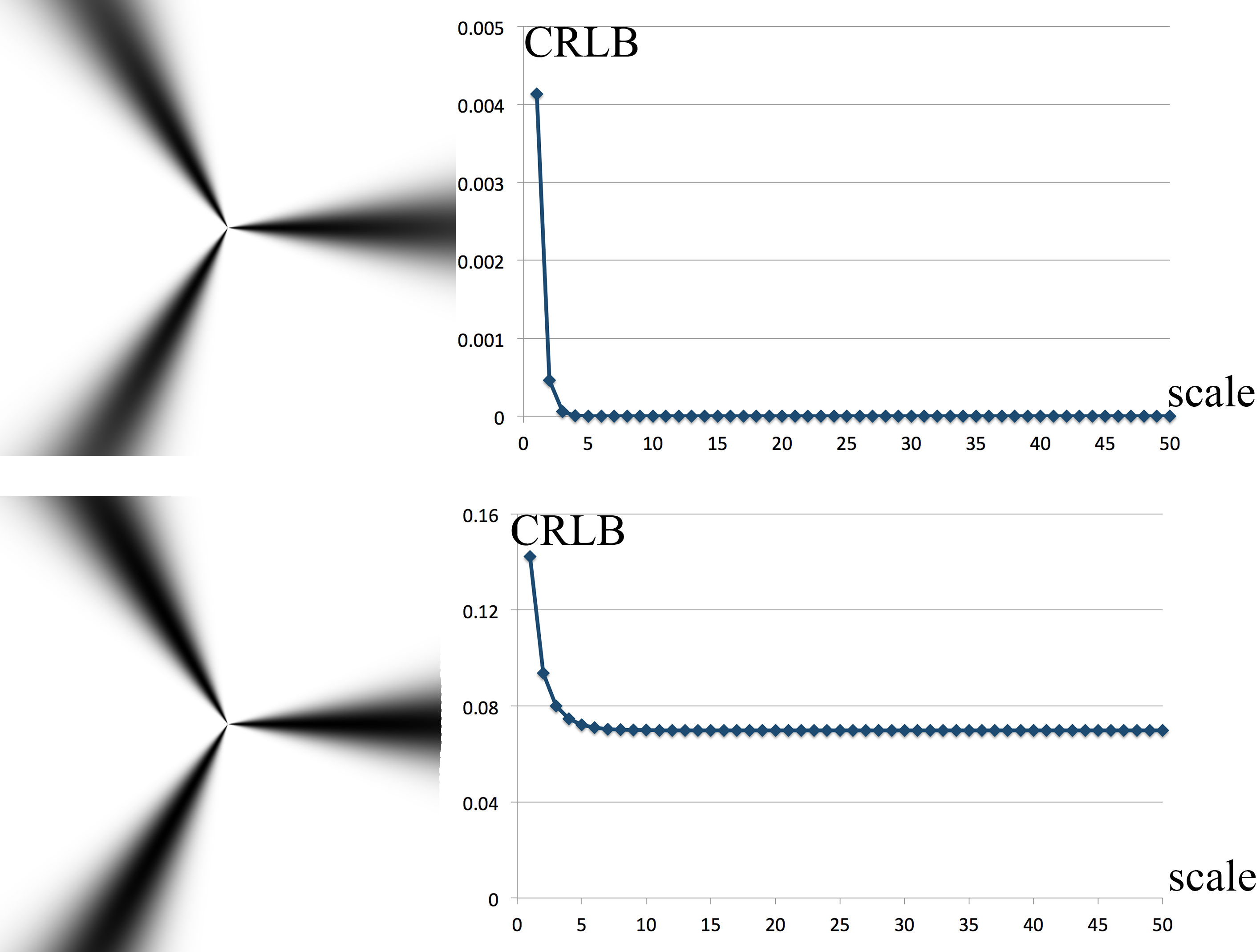}
}
	\caption{First column: Illustration of the analytically defined junction $J_2$ with $\lambda = 2.1$ and 4.5. Second column: The CRLB as a function of the number of the largest wavelet scale.}
	\label{im:uandun2num2}
\end{figure}

In this section, we compute the CRLB of the junction $\hat{J}_2$ \eqref{J2} with $\lambda = 2.1$ and 4.5. These choises of $\lambda$ correspond to two cases: first, where the background decays faster than the junction (in the Fourier domain); and second, where it is the junction that decays faster. The value of $\gamma$ is fixed as 2.5.

For the experiments, we have chosen the Meyer-type wavelet \eqref{eq:Meyer}. We again choose $\sigma_0 = 1$ since it provides a scaling factor that does not influence the rate of decay of the curve.

Based on the graphs of Figure \ref{im:uandun2num2}, we can observe the followings. For $\lambda = 2.1$, one wavelet scale was not enough for the CRLB to converge to a theoretical perfect estimation (like in Figure \ref{im:analytic3}, second experiment). However, by increasing the number of scales, the 0 bound is asymptotically achievable. This illustrates the main result of Proposition \ref{prop:asymptoticwavelet}, which shows that, by using finer scales, wavelets can improve the estimate, even in the case of junctions of a fixed size.
For $\lambda = 4.5$, as expected, the CRLB converges to a positive theoretical value.

\section{Conclusion}

In this paper, we considered the problem of estimating the orientations of features in images. In particular, we examined the orientation of patterns with rotational symmetry.  Within the framework of measurement functions composed of steerable filters, we derived Cram\'er-Rao lower bounds (CRLB) on the error of estimation.
We provided results on the connections to maximum likelihood estimation. 
Moreover, we discussed the problem of selecting the  parameters in the detector functions to achieve the lowest CRLB for a given reference template. 
In addition to the case of conventional detectors, we also studied the bounds on the performance of steerable wavelet estimators. We proposed an estimator for identifying orientations, and we compared its estimation error to the theoretical bounds. Finally, we provided several experiments on different realistic junctions and directional patterns that confirm the theory.

\section*{Acknowledgments}
The research leading to these results has received funding from the European Research Council under the European Union's Seventh Framework Programme (FP7/2007-2013)/ERC grant agreement $\text{n}^\circ$ 267439. The work was also supported by the Hasler Foundation.

\appendix

\subsection{Background Noise Model}
\label{sec:noise}

For our background signal $S$ to fulfil the requirements of self-similarity and isotropy, we define it as the (non-stationary) solution of a fractional stochastic differential equation. The general theory of such models and their non-Gaussian and sparse extensions is covered in \cite{UnserTaftiBook}. 

We assume that our background signal corresponds to a stochastic process $S$ on $\RR^2$ that is defined as the solution of the stochastic differential equation
\begin{align}
(-\Delta)^{\frac{\gamma}{2}} S = W,
\label{eq:Lsw}
 \end{align}
where $W$ is a 2D continuous-domain Gaussian white noise of variance $\sigma^2_0$ and $(-\Delta)^{\frac{\gamma}{2}}$ is the fractional Laplacian operator associated with the isotropic Fourier-domain multiplier $\norm{\V\omega}^{\gamma}$, with $\gamma \geq 0$. 

The intuitive idea here is to shape or ``color'' the white noise by an appropriately defined inverse fractional Laplacian, which gives to the solution an inverse-power-law spectrum.
The field we obtain is the isotropic 2D generalization of the fractional Brownian motion \cite{mandelbrot1968fractional}.

The observation $\langle S , f \rangle$ of the random process $S$ through a suitable test function $f$ is a Gaussian random variable.
The random process $S$ is Gaussian with $0$ mean (because $W$ is) and correlation form 
\begin{equation} \label{eq:correlfg}
	\mathcal{B}_S(f,g) = \mathbb{E} \left[ \langle S ,f \rangle \overline{\langle S , g \rangle} \right].
\end{equation}

The random variable $\langle W,  f \rangle$ is well-defined for $f \in L_2(\RR^2)$ \cite{UnserTaftiBook}. Therefore, $\langle S , f \rangle$ is well-defined when $(-\Delta)^{- \gamma / 2} f \in L_2(\RR^2)$, which is equivalent, in the Fourier domain, to $\int_{0}^{\infty} \omega^{1 - 2\gamma} \int_0^{2\pi} \lvert \hat{f}(\omega ,\varphi) \rvert^2 \mathrm{d} \varphi \mathrm{d} \omega < \infty$. 
This condition yields \eqref{eq:conditionxialpha} when $f = \xi_\alpha$. 
Assume that $f$ has vanishing moments until order $\lfloor \gamma \rfloor - 1$. This is, in particular, the case when $f = (-\Delta)^{\gamma /2} \phi$ with $\phi$ a smooth enough function in $L_1(\RR^2) \cap L_2(\RR^2)$ \cite[Proposition 5.1]{fageot2015wavelet}. Then, $\langle S, f\rangle$ is a well-defined Gaussian random variable with mean $0$ and variance $\sigma_0^2 \lVert (-\Delta)^{-\gamma/2} f \rVert_2^2 = \sigma_0^2 \lVert \phi \rVert_2^2$. The conditions on $f$ ensure that $(-\Delta)^{-\gamma/2} f$ is square-integrable. 
More generally, if $f$ and $g$ have enough vanishing moments, then the correlation form \eqref{eq:correlfg} is evaluated as
\begin{align}\label{eq:correlation} 
\mathcal{B}_S(f,g)
&=
\sigma_0^2 \langle (-\Delta)^{-\gamma / 2} f , \overline{(-\Delta)^{-\gamma/2} g} \rangle 
\\
&= \frac{\sigma_0^2}{(2\pi)^2} \int_0^{\infty}\int_0^{2\pi} \omega^{-2\gamma} \widehat{f} (\omega,\varphi) \overline{\widehat{g} ( \omega,\varphi)}  \mathrm{d}\varphi \omega \mathrm{d}\omega, \nonumber 
\end{align}
where we have used the Parseval relation and the polar Fourier coordinates for the last equality.
We interpret \eqref{eq:correlation} by saying that $S$ is an isotropic random field with generalized power spectrum $P_S(\boldsymbol{\omega}) = P_S(\omega) = \omega^{-2\gamma}$. 


\subsection{Complement on Steerability}
\label{sec:steerability}

We briefly recap some basic notions and facts on steerable functions.

\begin{definition}
  \label{def:steerable}
A function $\xi$ on the plane is steerable in the finite basis $\{\xi_\alpha\}$ if, for any rotation matrix $\M
R_{\theta_0}$, we can find coefficients $\{c_\alpha(\theta_0)\}$
such that
\begin{equation}
  \label{eq:steerabledef}
  \xi(\M R_{\theta_0}\V x) = \sum_{\alpha} c_{\alpha}(\theta_0) \xi_\alpha(\V x).
\end{equation}
A family of functions $\{\xi_\alpha\}$ is steerable if its linear span is invariant under arbitrary rotations.
\end{definition}

The last part of Definition \ref{def:steerable} is equivalent to saying that, for each $\theta_0$, a matrix $\M L(\theta_0)$ exists such that
\begin{equation}
  \label{eq:steerablematrix}
  \begin{pmatrix}
\xi_1(r,\theta+\theta_0) \\
\vdots\\
\xi_n(r,\theta+\theta_0) \\
  \end{pmatrix}
 = \M L(\theta_0)
   \begin{pmatrix}
\xi_1(r,\theta) \\
\vdots\\
\xi_n(r,\theta) \\
    \end{pmatrix}.
\end{equation}

An illustrative example of such a family is $\{\cos(\theta), \sin(\theta) \}$, whose rotations can be written as
\begin{equation}
\begin{pmatrix}
\cos(\theta+\theta_0) \\
\sin(\theta+\theta_0)
\end{pmatrix}
=
\begin{pmatrix}
\cos(\theta_0) &  -\sin(\theta_0)\\
\sin(\theta_0) &  \phantom{+}\cos(\theta_0)
\end{pmatrix}
\begin{pmatrix}
\cos(\theta)\\
\sin(\theta)
\end{pmatrix},
\end{equation}
which is a weighted sum of the unrotated functions.

If $\xi_\alpha$ is polar-separable as in \eqref{eq:polsep}, then its Fourier transform is also polar-separable as
\begin{align}
\hat{\xi}_\alpha(\omega,\varphi) = \hat{h}_\alpha(\omega) \ue^{\uj n_\alpha\varphi}
\label{eq:FourSepNbis}
\end{align}
for some $\hat{h}_\alpha$ related to $\eta_\alpha$ and $n_\alpha$. We note that this formulation provides a direct connection with the $n$th-order complex Riesz transform and leads to a comprehensive theory for the design and analysis of steerable filters and wavelets \cite{Felsberg,Koethe,NicolasSIAM,ZsSteerable}.


\subsection{Proof of Theorem \ref{prop:Statsq}}
\label{sec:proofGaussianvector}

The vector $\V q$ is Gaussian as a result of the ``integration" of the background noise $S$. As such, its law is characterized by its mean vector and covariance matrix. In the right part of \eqref{eq:tildetheta}, the first term is deterministic and the second is random with mean $0$. Hence, we have
\begin{align}
	\mathbb{E} [q_\alpha]  = \langle J ( \M R_{-\theta^*} \cdot ) , \xi_\alpha   \rangle  = \langle J  , \xi_\alpha ( \M R_{\theta^*}  \cdot ) \rangle  = \mathrm{e}^{\mathrm{j} n_\alpha \theta^*} \langle J, \xi_\alpha \rangle,
\end{align}
where we used \eqref{eq:rotatexialpha} for the last equality.

For the covariance matrix, we have that
\begin{align}
	\M C[\alpha, \beta] &=
	\mathbb{E} \left[ (q_\alpha - \mathbb{E} [q_\alpha])  \overline{(q_\beta - \mathbb{E} [q_\beta])} \right] \nonumber \\
	& = \mathbb{E} \left[ \langle S ,\xi_\alpha \rangle \overline{\langle S , \xi_\beta \rangle} \right].
\end{align}
Then, applying \eqref{eq:correlation} to $f = \xi_\alpha$ and $g =\xi_\beta$, we deduce that
\begin{align}
	\M C[\alpha, \beta] = \frac{\sigma^2_0}{(2\pi)^2} \left( \int_0^\infty w^{1 - 2\gamma} \widehat{h_\alpha}(\omega) \overline{\widehat{h}_\beta (\omega)} \mathrm{d} \omega \right) \delta[n_\alpha - n_\beta],
\end{align}
as expected, where we have exploited the Fourier domain expression \eqref{eq:FourSepNbis}. 

Finally, the Fisher information of a Gaussian model is provided by the Slepian-Bangs formula \cite[B.3.3]{StoicaSpectral}. In our case, however, we deal with \textit{complex} Gaussian vectors, because the measurement functions are typically complex (see \eqref{eq:polsep}). We therefore use a generalization of the Slepian-Bangs formula adapted to complex Gaussian random vectors \cite[B.3.25]{StoicaSpectral}. Moreover, this formula simplifies because the covariance matrix does not depend on the parameter $\theta^*$, giving \eqref{eq:FIq}. 


\subsection{Proof of Theorem \ref{Th:distinctH}}
\label{sec:distinctHproof}

As already noted, the CRLBs based on the sets $H$ and $H_+$ are the same.
We therefore restrict the decision to positive harmonics $n \in H_+$.

The functions $\xi_n$ are assumed to have distinct harmonics; therefore, the matrix \eqref{eq:covq} is diagonal with
\begin{equation} \label{eq:diago}
	\M C[n,n] = \frac{\sigma^2_0}{(2\pi)^2} \int_0^\infty \omega^{1-2\gamma} \lvert \hat{h}_n(\omega)\rvert^2 \mathrm{d}\omega. 
\end{equation}
Knowing that $\M C$ is diagonal, we easily deduce, starting from \eqref{eq:FIq}, that the Fisher information is in this case
\begin{align}
	\mathrm{FI}(\theta^*)  &=  \sum_{n\in H^+} (\mathrm{j} n ) \mathrm{e}^{\mathrm{j} n  \theta^*} \overline{u_n}  (\M C [n,n])^{-1} (- \mathrm{j} n ) \mathrm{e}^{- \mathrm{j} n \theta^*} u_n \nonumber \\
	& = \sum_{n \in H_+} n^2 \lvert u_n\rvert^2 ( \M C [n,n])^{-1}.
\end{align}
Knowing $C[n,n]$, we deduce the Fisher information, and hence the CRLB which is the inverse of the Fisher information. 
We finally remark that the CRLB does not depend on $\theta^*$ in that case.

\subsection{Construction of the Tridiagonal Covariance Matrix}
\label{sec:covmat}
The measurements are rearranged in the following way to create a (at most) tridiagonal covariance matrix. 
Our starting point is an arbitrary set of unique measurement indices of the form $(n,i)$, where $n$ is the harmonic and $i$ is the scale of the corresponding measurement. As before, since $(n,i)$ and $(-n,i)$ carry the same information for real-valued patterns, with $q_{-n,i} = \overline{q_{n,i}}$ and $u_{-n,i} = \overline{u_{n,i}}$, we can assume that all harmonics $n$ are positive without loss of generality.
We divide the set of measurement indices into as few disjoint sets $G_r$ as possible, subject to three conditions ($r=1,\ldots,g$ is an index for the sets).
\begin{itemize}
  \item Each set $G_r$ consists of indices with a single fixed harmonic $n_r$.
  \item Each set $G_r$ only contains indices with consecutive scales---it can also contain only a single element.
  \item If the two sets $G_r,G_s$ share the same harmonic $n_r=n_s$, then the scales in $G_r,G_s$ differ by a minimum of $2$.
\end{itemize}
The last condition is a consequence of the first two and requires one to have a minimal number of sets. Here, $g$ denotes the total number of such sets.

From the above conditions, it follows that each set $G_r$ has some $l_r$ elements of the form $(n_r,i_r),\ldots,(n_r,i_r+l_r-1)$. We use the notation \eqref{eq:rle} and re-index all measurements as $\tilde{q}_{r|e}$. The idea behind this re-indexing is that, following \eqref{eq:waveletcov}, the covariance matrix of the measurements $\{\tilde{q}_{r|e}\}$ for each fixed $r$ is Toepliz-tridiagonal of the form \eqref{eq:Tr}. Moreover, the overall covariance matrix is block-diagonal, with $\M T_r$, $r=1,\ldots,g$, as its diagonal blocks.


\subsection{Proof of Theorem \ref{Th:WaveletCRLB}.}
\label{proof:th2}
We first start with a lemma that shows that the lower bound in \eqref{CRLBwavelet} is strictly positive, and that will be used in the proof of Theorem \ref{Th:WaveletCRLB}.

\begin{lemma}
Assume that $\hat{h}$ is real-valued and that $\hat{h}(\V\omega) = 0$ for $\V\omega \notin (\pi/4, \pi ]$. Then, for every $z$, we have that
$ b_z > 2^{z/2+2}\abs{d_z}.$
Consequently, we have the relation
\begin{align}
B > 2 |D|.
\label{eq:B2D}
\end{align}
\label{lemma:ab}
\end{lemma}
%
\begin{proof}
We have that
\begin{align}
\abs{d_z} 
&= 
\frac{1}{2\pi} \left|\int_{\pi/4}^{\pi/2} \omega^z \hat{h}(\omega) \hat{h}(2\omega) \omega \ud\omega\right| \notag \\
&\leq 
\frac{1}{2\pi} \sqrt{\int_{\pi/4}^{\pi/2} \omega^z \hat{h}(\omega)^2  \omega \ud\omega \int_{\pi/4}^{\pi/2} \omega^z \hat{h}(2\omega)^2  \omega \ud\omega}  \notag \\
&= 
\frac{1}{2\pi} \frac{1}{2^{z/2+1}} \sqrt{ \int_{\pi/4}^{\pi/2} \omega^z \hat{h}(\omega)^2  \omega \ud\omega \int_{\pi/2}^{\pi} \omega^z \hat{h}(\omega)^2 \omega \ud\omega } \notag \\
&\leq 
\frac{1}{2\pi} \frac{1}{2^{z/2+2}} { \left( \int_{\pi/4}^{\pi/2} \omega^z \hat{h}(\omega)^2  \omega \ud\omega + \int_{\pi/2}^{\pi} \omega^z \hat{h}(\omega)^2 \omega \ud\omega \right) } \nonumber \\
&= 
\frac{1}{2^{z/2+2}} b_z,
\label{eq:dzbz}
\end{align}
where we have used the Cauchy-Schwarz inequality, and the inequality of arithmetic and geometric means.
By checking the equality conditions of the two inequalities, we see that they do not happen for the cases of interest.
 The equality happens only if $\hat{h}(\omega) = \hat{h}(2\omega)$ almost everywhere, which is not the case here. Thus, the  inequality in \eqref{eq:dzbz} is strict. 
Finally, by selecting $z = -2\gamma$, we deduce \eqref{eq:B2D}.
\end{proof}

Let us now prove Theorem \ref{Th:WaveletCRLB}.
The estimation of ${\theta^*}$ from $\{q_{n,i}\}$ is essentially the same as the estimation from $\{\tilde{q}_{r|e}\}$. In particular, the CRLB is the same.

The measurement vector $\tilde{\V q}$ is constructed by concatenating  the vectors $\tilde{\V q}_r = (\tilde q_{r|1},\ldots,\tilde q_{r|l_r})$, $r=1,\ldots,g$. This is a normal vector, with its mean $\tilde{\V\mu}$ given by concatenating the vectors $\tilde{\V\mu}_r = (\ue^{\uj n_r\theta^*}2^{-i_r\gamma}u_{r|1},\ldots,\ue^{\uj n_r\theta^*}2^{-(i_r+l_r-1)\gamma}u_{r|l_r})$, $r=1,\ldots,g$.

The covariance of $\tilde{\V q}$, as already noted, is block-diagonal of the form
\begin{align}
 \tilde{\V C} =
\begin{pmatrix}
    \M T_{1} & 0  & \dots  & 0 \\
    0 & \M T_{2}  & \dots  & 0 \\
    \vdots & \vdots  & \ddots & \vdots \\
    0 & 0 & \dots  & \M T_{g}
\end{pmatrix},
\end{align}
where the form of $\M T_r$ is given by \eqref{eq:Tr}. 
Each $\M T_r$ is a Toeplitz-tridiagonal matrix (or the scalar $B$ if $l_r=1$). Consequently, its eigenvalues are given by
\begin{align}
\lambda_r^{(t)} = B+2D\cos\left( \frac{t\pi}{l_r+1} \right),
\end{align}
where $1\leq t \leq l_r$. The corresponding eigenvectors are
\begin{align}
\V v_r^{(t)} = \sqrt{\frac{2}{k+1}}  \begin{pmatrix}
\sin\left(\frac{t\pi}{l_r+1}\right) \\
\sin\left(2\frac{t\pi}{l_r+1}\right) \\
\vdots \\
\sin\left(k\frac{t\pi}{l_r+1}\right)
\end{pmatrix}.
\end{align}
For positive $n_r$ (which we can assume without loss of generality), we have, according to \eqref{eq:FIq}, 
\begin{align}
\mathrm{FI}({\theta^*})
&= 2\Re\left( \left( \frac{\ud}{\ud \theta} \tilde{\V \mu} \right)^\HH {\tilde{\V C}}^{-1} \left( \frac{\ud}{\ud \theta} \tilde{\V\mu} \right)\right) \notag\\
&= 2\sum_{r=1}^g n_r^2 \tilde{\tilde{\V \mu}}_r^\HH \M T_{r}^{-1} \tilde{\tilde{\V\mu}}_r,
\end{align}
where
\begin{equation}
\tilde{\tilde{\V\mu}}_r = ( 2^{-i_r\gamma} u_{r | 1}, \ldots, 2^{-(i_r+l_r-1)\gamma} u_{r |  l_r} ).
\end{equation}
Thus, the explicit formula for the Fisher information is 
\begin{align}
& \mathrm{FI}({\theta^*}) ={}2\sum_{r=1}^g n_r^2 \sum_{t=1}^{l_r} \left( B +2D\cos\left(\frac{t\pi}{l_r+1}\right) \right)^{-1} \notag\\
& {}\times{} \abs{ \sum_{e=1}^{l_r} 2^{-(i_r+e-1)\gamma} u_{r |  e} \sin \left(e\frac{t\pi}{l_r+1} \right) }^2.
\end{align}
Also, 
\begin{align}
\frac{2}{B+2\abs{D}} \sum_{r=1}^g n_r^2 \norm{ \tilde{\tilde{\mu}}_r }_2^2 \leq \mathrm{FI}({\theta^*}) \leq \frac{2}{B-2\abs{D}} \sum_{r=1}^g n_r^2  \norm{ \tilde{\tilde{\mu}}_r }_2^2,
\end{align}
where we have used the positivity of $B-2|D|$ demonstrated in  Lemma \ref{lemma:ab}.
This shows that $\mathrm{FI}({\theta^*})$ is finite if and only if $\sum_{r=1}^g n_r^2 \norm{ \tilde{\tilde{\mu}}_r }_2^2 $ is finite.
Going back to the original indices $(n,i)$, the CRLB therefore satisfies \eqref{CRLBwavelet}.

\subsection{Proof of Proposition \ref{prop:asymptoticwavelet}.}
\label{pr:gen}
The first equivalence directly follows from Theorem \ref{Th:WaveletCRLB} and the fact that $B-2|D|>0$ (Lemma \ref{lemma:ab}).
We now show that the series $\sum_{n,i} n^2 4^{-i\gamma} \abs{u_{n,i}}^2 $ converges if and only if $(-\Delta)^{\gamma/2} \frac{\partial}{\partial \theta} J(r,\theta) \in L_2(\mathbb{R}^2)$, implying the second equivalence.
By assumption on $J$, we write
\begin{align}
J(r,\theta) = \sum_{(n,i)\in\ZZ^2} v_{n,i} \xi_{n,i}(r,\theta).
\end{align}
We recall that, according to (53), $ \langle \partial_\theta J,\xi_n\rangle = -\langle J,\partial_\theta\xi_n\rangle$.
Therefore, 
\begin{align}
&\norm{ (-\Delta)^{\gamma/2} \frac{\partial}{\partial \theta} J(r,\theta) }_2^2 \notag \\
&= \norm{ \sum_{(n,i)\in\ZZ^2} v_{n,i} (-\Delta)^{\gamma/2} \frac{\partial}{\partial \theta} \xi_{n,i}(r,\theta) }_2^2 \notag \\
&= \norm{ \sum_{(n,i)\in\ZZ^2} \uj n v_{n,i} (-\Delta)^{\gamma/2} \xi_{n,i}(r,\theta) }_2^2 \notag \\
&= \!\!\!\!\!\!\sum_{n_1,n_2, i_1, i_2}\!\!\!\!\!\! n_1 n_2 v_{n_1,i_1} v^*_{n_2,i_2} \bigl\langle  (-\Delta)^{\gamma/2} \xi_{n_1,i_1},  (-\Delta)^{\gamma/2} \xi_{n_2,i_2} \bigr\rangle .
\end{align}
With the same polar coordinate computation as in \eqref{eq:correlation}, we obtain that $\bigl\langle  (-\Delta)^{\gamma/2} \xi_{n_1,i_1},  (-\Delta)^{\gamma/2} \xi_{n_2,i_2} \bigr\rangle$ is given by
\begin{align}
\frac{2^{\abs{i_1 -i_2}-2\min(i_1,i_2)\gamma}}{2\pi} \delta[n_1-n_2] 
\int_0^{\infty} \omega^{2\gamma} \hat{h}(\omega) \hat{h}(2^{\abs{i_1 -i_2}}\omega) \omega\ud\omega.
\end{align}
This can be simplified as
\begin{align}
&\bigl\langle  (-\Delta)^{\gamma/2} \xi_{n_1,i_1},  (-\Delta)^{\gamma/2} \xi_{n_2,i_2} \bigr\rangle \notag \\
&= \begin{cases}
 4^{-i\gamma} b_{-2\gamma}, & \text{if } n_1=n_2, i_1=i_2 \\
 2\cdot4^{-\min(i_1,i_2)\gamma} d_{-2\gamma} , & \text{if } n_1=n_2, \abs{i_1-i_2} =1 \\
0, & \text{otherwise}.
\end{cases}
\end{align}
Therefore,
\begin{align}
&\norm{ (-\Delta)^{\gamma/2}\frac{\partial}{\partial \theta} J(r,\theta) }_2^2 = \sum_{n\in\ZZ} n^2  \Bigl( \sum_{i\in \ZZ}  \bigl( 4^{-i\gamma} \abs{v_{n,i}}^2 b_{-2\gamma}  \notag\\
& {} + 2\cdot4^{-i\gamma}v_{n,i}^{*} v_{n,i+1}d_{-2\gamma} + 2\cdot4^{-(i-1)\gamma} v_{n,i}^{*} v_{n,i-1}d_{-2\gamma}  \bigr) \Bigr).
\end{align}
By employing the Cauchy-Schwarz inequality,
\begin{align}
&\abs{\sum_i 4^{-i\gamma} v^*_{n,i} v_{n,i+1}} \notag\\
&\leq \sqrt{\left( \sum_i 4^{-i\gamma} \abs{v_{n,i}}^2 \right)\left( \sum_i 4^{-i\gamma} \abs{v_{n,i+1}}^2 \right)} \notag\\
&= 2^{\gamma} \sum_i 4^{-i\gamma} \abs{v_{n,i}}^2.
\end{align}
Also, we deduce
\begin{align}
&\abs{\sum_i 4^{-(i-1)\gamma} v^*_{n,i} v_{n,i-1}} \notag\\
&\leq \sqrt{\left( \sum_i 4^{-(i-1)\gamma} \abs{v_{n,i}}^2 \right)\left( \sum_i 4^{-(i-1)\gamma} \abs{v_{n,i-1}}^2 \right)} \notag\\
&= 2^{\gamma} \sum_i 4^{-i\gamma} \abs{v_{n,i}}^2.
\end{align}
We arrive at the inequalities
\begin{align} \label{eq:controlJnormtofinish}
&\left( b_{-2\gamma} -2^{- \gamma+2}|d_{-2\gamma}| \right)\left( \sum_{(n,i)\in \ZZ^2} n^2 4^{-i\gamma} \abs{v_{n,i}}^2 \right)\notag\\
&\quad{} \leq \norm{ (-\Delta)^{\gamma/2}\frac{\partial}{\partial \theta} J(r,\theta) }_2^2 \notag \\
&\quad{}\leq \left( b_{-2\gamma} + 2^{- \gamma+2}|d_{-2\gamma}| \right)\left( \sum_{(n,i)\in \ZZ^2} n^2 4^{-i\gamma} \abs{v_{n,i}}^2 \right).
\end{align}
Lemma \ref{lemma:ab} implies that $b_{-2\gamma} > 2^{- \gamma+2}\abs{d_{-2\gamma}}$. Hence, \eqref{eq:controlJnormtofinish} means that $(-\Delta)^{\gamma/2}\frac{\partial}{\partial \theta} J(r,\theta) \in L_2(\mathbb{R}^2)$ if and only if $\sum_{(n,i) \in \ZZ^2} n^2 4^{-i\gamma} \abs{v_{n,i}}^2 < \infty$.

The remaining part is to show that $\sum_{(n,i) \in \ZZ^2} n^2 4^{-i\gamma} \abs{v_{n,i}}^2 < \infty$ if and only if $\sum_{(n,i) \in \ZZ^2} n^2 4^{-i\gamma} \abs{u_{n,i}}^2 < \infty$, from which we deduce the second part of Proposition \ref{prop:asymptoticwavelet}. 
We therefore focus on the relationship between $u_{n,i}$ and $v_{n,i}$ and remark that
\begin{align}
u_{n,i} = \left\langle J(r,\theta), \xi_{n,i} \right\rangle = \sum_{{n'},{i'}} v_{{n'},{i'}} \left\langle \xi_{{n'},{i'}}, \xi_{n,i} \right\rangle.
\end{align}
Again, it is straightforward to verify that
\begin{align}
\left\langle \xi_{{n'},{i'}}, \xi_{i,n} \right\rangle = \begin{cases}
b_0, & \text{if } n={n'}, i={i'} \\
2 d_0, & \text{if } n={n'}, \left\vert i-{i'}\right\vert  = 1 \\
0, & \text{otherwise}.
 \end{cases}
\end{align}
Hence,
\begin{align}
u_{n,i} = 2d_0 v_{n,i-1} + b_0v_{n,i} + 2d_0 v_{n,i+1}
\end{align}
and
\begin{align}
&2^{-i\gamma} u_{n,i} = \left(2^{1-\gamma} d_0 \right)\left(  2^{-(i-1)\gamma} v_{n,i-1} \right) + b_0 \left( 2^{-i\gamma} v_{n,i} \right) \notag\\
&+ \left( 2^{1+\gamma} d_0 \right) \left( 2^{-(i+1)\gamma} v_{n,i+1} \right).
\end{align}
This indicates that, for any given $n$, the sequence $\left\{ 2^{-i\gamma} u_{n,i} \right\}_{i \in \ZZ}$ is the result of applying to the sequence $\left\{ 2^{-i\gamma} v_{n,i} \right\}_i$ the FIR filter with $z$ transform
\begin{align}
F(z) = 2^{1-\gamma} d_0 z^{-1} + b_0 + 2^{1+\gamma} d_0 z.
\end{align}
This filter has no poles on the unit circle. Moreover, if
\begin{align} \label{eq:conditiongamma}
\frac{b_0}{\abs{d_0}} \neq 2^{1+\gamma} + 2^{1-\gamma},
\end{align}
then none of its zeros will lie on the unit circle. Thus,
\begin{align}
0 < \min_{\omega \in [0,2\pi) } \abs{F\left(\ue^{\uj \omega}\right) }^2 \leq \max_{\omega \in [0,2\pi) }  \abs{F\left(\ue^{\uj \omega}\right) }^2 < \infty.
\end{align}
The filtering relationship between the two sequences implies the following inequalities on their energies:
\begin{align}
&\left(  \min_{\omega \in [0,2\pi) } \abs{F\left(\ue^{\uj \omega}\right) }^2  \right) \sum_{i \in \ZZ} \abs{2^{-i\gamma} v_{n,i}}^2 \notag\\
&\leq \sum_{i \in \ZZ}  \abs{2^{-i\gamma} u_{n,i}}^2 \notag\\
&\leq \left(  \max_{\omega \in [0,2\pi) }  \abs{F\left(\ue^{\uj \omega}\right) }^2 \right) \sum_{i \in \ZZ} \abs{2^{-i\gamma} v_{n,i}}^2.
\end{align}
Thus,
\begin{align}
&\left(  \min_{\omega \in [0,2\pi) } \abs{F\left(\ue^{\uj \omega}\right) }^2  \right) \sum_{(i,n)\in \ZZ^2} n^2 4^{-i\gamma} \abs{ v_{n,i} }^2 \notag \\
&\leq \sum_{(i,n)\in \ZZ^2} n^2 4^{-i\gamma} \abs{ u_{n,i} }^2 \notag\\
& \leq \left(  \max_{\omega \in [0,2\pi) } \abs{F\left(\ue^{\uj \omega}\right) }^2  \right) \sum_{(i,n)\in \ZZ^2} n^2 4^{-i\gamma} \abs{ v_{n,i} }^2.
\end{align}
Consequently,
\begin{align}
\sum_{n,i} n^2 4^{-i\gamma} \left\vert v_{n,i}\right\vert ^2 < \infty  \Leftrightarrow \sum_{n,i} n^2 4^{-i\gamma} \left\vert u_{n,i}\right\vert ^2 < \infty,
\end{align}
as expected, and the proof is finished.

We finally remark that value of $\gamma$ that is excluded in the proof satisfies \eqref{eq:conditiongamma}. We provide the corresponding values for different radial profiles in Table \ref{tab:excludedGamma}. 
\begin{table}[!t]
\centering
\caption{Excluded $\gamma$ values}
\vspace{0.05 in}
\begin{tabular}{c | l}
\hline
\hline
Wavelet type & Excluded $\gamma$\\
\hline
Shannon  \cite{DNU} & $
\infty
$\\
Simoncelli \cite{PortillaParametric} & $
0.736
$\\
Meyer \cite{daubechies1992ten} & $
1.978 $\\
Papadakis \cite{Papadakis} & $
1.476 $\\
Held  \cite{held2010steerable}  & $
1.443  $\\
\hline
\hline
\end{tabular}
\label{tab:excludedGamma}
\end{table}


\bibliographystyle{ieeetr}
\bibliography{refs_cramer}

\end{document}